\documentclass[11pt]{article}

\usepackage[T1]{fontenc}
\usepackage[utf8]{inputenc}
\usepackage{amsmath,amssymb,amsthm,mathtools,bm}
\usepackage{microtype}
\usepackage{geometry}
\usepackage{hyperref}
\usepackage{enumitem}
\usepackage{booktabs}
\usepackage{array}
\usepackage{authblk}
\hypersetup{colorlinks=true,citecolor=blue,linkcolor=blue,urlcolor=blue}

\newtheorem{theorem}{Theorem}
\newcommand{\Tr}{\operatorname{Tr}}
\newcommand{\dd}{\mathrm{d}}
\newcommand{\KL}{D_{\mathrm{KL}}}
\newcommand{\Tcl}{T_{\mathrm{cl}}}
\newcommand{\Tq}{T_{\mathrm{q}}}
\newcommand{\GQD}{D_{\mathrm{GQD}}}

\title{Sklar's Theorem and Quantum State Reconstruction from All-Context Dependence}

\author{Yi-Yu Lin$^{1,2}$\thanks{yiyu@simis.cn}}

\affil{${}^1$Fudan Center for Mathematics and Interdisciplinary Study, Fudan University, Shanghai, 200433, China}
\affil{${}^2$Shanghai Institute for Mathematics and Interdisciplinary Sciences (SIMIS), Shanghai, 200433, China}
\date{}
\begin{document}
\maketitle
\begin{abstract}
Sklar's theorem separates the marginals of a joint distribution from its dependence structure. We apply this viewpoint to Born statistics generated by local quantum measurements. For non-product two-qubit states, the dependence nuclei over all local binary projective measurement contexts determine the state up to at most a discrete double-spin-flip ambiguity. The corresponding scalar optimization, combined with the quantum total correlation, reproduces global quantum discord.
\end{abstract}

\section{Introduction}

In classical probability theory, Sklar's theorem states that a joint distribution can be decomposed into its marginal distributions and a copula that captures their dependence structure. For continuous random variables, the copula is uniquely determined by the joint distribution, providing a natural separation between ``how each variable is distributed on its own'' and ``how the variables depend on one another'' \cite{Sklar1959,Nelsen2006}.

Quantum mechanics immediately supplies a class of probability distributions that can be reconsidered from this viewpoint. Given a multipartite quantum state and a choice of local measurement bases, the Born rule produces a classical joint distribution for the measurement outcomes. If the measurement is repeated experimentally, what is obtained is precisely an ordinary joint statistical table. Once the measurement bases have been fixed, the probability table itself has no special ``quantum grammar'': one may ask of it the same questions about marginals, correlations, and dependence as for classical joint data arising elsewhere. It is therefore natural to ask the same question suggested by Sklar's theorem: after removing the marginal statistics of the local outcomes, what dependence structure remains?

The quantum setting, however, contains an additional structure that is usually absent from an ordinary classical-statistical problem. The same quantum state produces different joint distributions in different local measurement contexts. We therefore do not examine the dependence structure in one fixed context only, but the entire family of dependence structures generated by all local contexts. This leads to a simple question: \emph{if the marginal information is discarded in every local measurement context, how much information about the original quantum state remains in the collection of all resulting dependence structures?}

Formally, we study the map
\begin{equation}
\rho \longmapsto \left\{\mathcal D\!\left(p_{\rho,B}\right)\right\}_{B},
\label{eq:mastermap}
\end{equation}
where $p_{\rho,B}$ is the joint distribution generated from the state $\rho$ by the Born rule in a local measurement context $B$, and $\mathcal D$ denotes the dependence structure extracted from that distribution after the marginal information has been removed.

The basic structure revealed by Sklar's theorem is the separation of marginals from dependence. When the marginal distributions are continuous, this separation is especially clean: the joint distribution determines a unique copula, so the copula itself can serve as the margin-free dependence object. For discrete variables, copula-based multivariate modeling is also well developed \cite{Nikoloulopoulos2013,Panagiotelis2012}. For the present purpose, however, we are interested in a more structural question: given a discrete joint table that is already fixed, how should one extract an intrinsic dependence object that is independent of its marginals? A subtlety is that, for discrete marginals, the copula in Sklar's representation is generally not uniquely determined by the joint distribution, so the correspondence between ``the copula'' and intrinsic margin-free dependence is no longer as direct as in the continuous case \cite{GenestNeslehova2007}. We therefore adopt the dependence-nucleus viewpoint developed by Geenens \cite{Geenens2020}: probability tables related only by positive local reweightings are regarded as having the same dependence, and the resulting equivalence class is called a dependence nucleus. Thus, in the terminology of this paper, the copula is the object that carries margin-free dependence in the continuous case, whereas the dependence nucleus is the corresponding dependence object that we adopt in the discrete case.

The rest of the paper is organized as follows. Section~\ref{sec:reconstruction} introduces the dependence nucleus for binary discrete distributions and proves the all-context reconstruction theorem for two qubits. Section~\ref{sec:discord} compresses the full dependence structure in each context to a scalar and shows that the corresponding optimization reproduces global quantum discord. Section~\ref{sec:outlook} discusses natural extensions to multiqubit and continuous-variable systems. Appendix~\ref{app:sklar} reviews the basic facts about Sklar's theorem and copulas used in the main text.

\section{All-context dependence reconstruction for two qubits}
\label{sec:reconstruction}

\subsection{Dependence structure in binary discrete distributions}
\label{subsec:binary}

Sklar's theorem, the general definition of a copula, and the continuous-variable case are reviewed in Appendix~\ref{app:sklar}. Here we introduce only the binary discrete structure needed for the two-qubit analysis.

In the continuous case, a copula separates the marginal statistics from the dependence structure of a joint distribution; equivalently, one can transform each marginal to a uniform distribution while retaining only the dependence. For discrete random variables, the ordinary Sklar copula is no longer uniquely determined by the joint distribution. We therefore use the dependence nucleus of Ref.~\cite{Geenens2020}, defined by transformations that change local marginal weights while preserving the dependence structure.

Consider two binary random variables $X,Y\in\{+,-\}$ with a strictly positive joint distribution
\begin{equation}
p=
\begin{pmatrix}
 p_{++} & p_{+-}\\
 p_{-+} & p_{--}
\end{pmatrix},
\qquad
p_{st}>0,
\qquad
\sum_{s,t=\pm}p_{st}=1.
\end{equation}
Two probability tables are said to belong to the same dependence nucleus if there exist positive numbers $r_s,c_t$ such that
\begin{equation}
p'_{st}=\frac{1}{Z}\,r_s c_t\,p_{st},
\label{eq:reweight}
\end{equation}
where $Z$ is a normalization factor. The factor $r_s$ depends only on the value of $X$, while $c_t$ depends only on the value of $Y$; the transformation can therefore alter the two marginal distributions without changing their joint dependence in this sense.

The equivalence relation in Eq.~\eqref{eq:reweight} itself remains meaningful on the boundary of the probability simplex, where some cells may vanish: positive local reweightings preserve the support pattern of the table. In that case the support pattern is part of the dependence nucleus, and the ordinary odds ratio need not be a finite or complete coordinate. The odds-ratio description below therefore applies to strictly positive tables, while the reconstruction theorem in Sec.~\ref{subsec:twoqubit} will use Eq.~\eqref{eq:reweight} directly so that rank-deficient quantum states are included.

For a strictly positive $2\times2$ table, the dependence nucleus is completely characterized by the odds ratio
\begin{equation}
\omega(p)=\frac{p_{++}p_{--}}{p_{+-}p_{-+}}.
\label{eq:odds-basic}
\end{equation}
Indeed,
\begin{equation}
\omega(p')=
\frac{(r_+c_+p_{++})(r_-c_-p_{--})}
{(r_+c_-p_{+-})(r_-c_+p_{-+})}
=\omega(p),
\end{equation}
and conversely two strictly positive $2\times2$ tables with the same odds ratio belong to the same dependence nucleus. Thus $\omega$ is a complete coordinate for the margin-free dependence structure in the binary case.

Each such dependence nucleus also contains a unique representative with uniform marginals,
\begin{equation}
\overline p_{\omega}
=\frac{1}{2(1+\sqrt\omega)}
\begin{pmatrix}
\sqrt\omega & 1\\
1 & \sqrt\omega
\end{pmatrix}.
\label{eq:uniformrep}
\end{equation}
It satisfies
\begin{equation}
\overline p_X(+)=\overline p_X(-)=\overline p_Y(+)=\overline p_Y(-)=\frac12.
\end{equation}
This uniform-margin representative is the discrete counterpart of the copula representation in the continuous case. In the latter, the probability-integral transform
\begin{equation}
U=F_X(X),\qquad V=F_Y(Y)
\end{equation}
makes the two marginals uniform on $[0,1]$ while leaving the dependence between $U$ and $V$; in the binary case, $\overline p_\omega$ likewise fixes both marginals to be uniform while retaining only the dependence encoded by $\omega$.

Independence corresponds to
\begin{equation}
\omega=1,
\end{equation}
for which
\begin{equation}
\overline p_{\omega=1}=\frac14
\begin{pmatrix}
1&1\\1&1
\end{pmatrix}.
\end{equation}
Thus the three structures needed below may be read in parallel as
\begin{equation}
\begin{array}{ccc}
\text{continuous case} && \text{binary discrete case}\\[2pt]
\text{margin-free copula} &\longleftrightarrow& \text{dependence nucleus}\\
\text{uniform marginals} &\longleftrightarrow& \text{uniform-margin representative}\\
\text{dependence structure} &\longleftrightarrow& \omega.
\end{array}
\end{equation}
We now apply this construction directly to Born probability tables produced by two-qubit states in arbitrary local measurement bases.

\subsection{All-context dependence reconstruction}
\label{subsec:twoqubit}

Once local measurement bases are chosen, the Born rule turns a quantum state into an ordinary classical joint probability table. The preceding subsection shows that, for a strictly positive binary table, its margin-free dependence is completely characterized by the odds ratio. We now ask whether the corresponding dependence data, collected over all local measurement contexts, suffice to determine the underlying two-qubit state.

For later convenience, we use the bounded reparametrization
\begin{equation}
Q
:=
\frac{
p_{++}p_{--}-p_{+-}p_{-+}
}{
p_{++}p_{--}+p_{+-}p_{-+}
}.
\label{eq:Qdef}
\end{equation}
For a strictly positive table this is simply
\begin{equation}
Q=\frac{\omega-1}{\omega+1},
\label{eq:Qomega}
\end{equation}
so it carries exactly the same dependence information as the odds ratio. Under a local positive reweighting
\begin{equation}
p_{st}\longmapsto
\frac{1}{Z}r_s c_t\,p_{st},
\qquad
r_s,c_t>0,
\end{equation}
the two products in Eq.~\eqref{eq:Qdef} acquire the same positive multiplicative factor. Hence $Q$ is invariant under the same local reweightings that remove the marginal information.

The advantage of $Q$ is that it remains finite at the boundary of the probability simplex. We shall show below that, for every non-product two-qubit state, its denominator is strictly positive in every local projective measurement context. Thus $Q$ is well defined even for pure entangled states and rank-deficient mixed states.

\begin{theorem}[All-context dependence reconstruction for two qubits]
Let $\rho$ and $\sigma$ be two non-product two-qubit states. Suppose that, for every pair of local binary projective measurements,
\begin{equation}
Q_\rho(\mathbf x,\mathbf y)
=
Q_\sigma(\mathbf x,\mathbf y)
\qquad
\forall\,\mathbf x,\mathbf y\in S^2.
\label{eq:allcontext-Q}
\end{equation}
Then
\begin{equation}
\sigma=\rho
\quad\text{or}\quad
\sigma=\widetilde\rho,
\end{equation}
where
\begin{equation}
\widetilde\rho
=
(\sigma_y\otimes\sigma_y)
\rho^*
(\sigma_y\otimes\sigma_y)
\label{eq:spinflip}
\end{equation}
is the double spin flip of $\rho$ \cite{Wootters1998}.
\footnote{
For fixed local measurement directions $(\mathbf x,\mathbf y)$,
\[
p_{st}^{\widetilde\rho}(\mathbf x,\mathbf y)
=
p_{-s,-t}^{\rho}(\mathbf x,\mathbf y).
\]
Thus the double spin flip simultaneously exchanges the two binary outcome labels on both sides. The products
\[
p_{++}p_{--},
\qquad
p_{+-}p_{-+}
\]
are separately unchanged, and hence so is $Q$. This ambiguity does not imply that $\rho$ and $\widetilde\rho$ are physically identical as quantum states. See Ref.~\cite{Wootters1998} for the standard two-qubit spin flip.
}
Hence every non-product two-qubit state, including pure entangled states and rank-deficient mixed states, is reconstructible from its all-context margin-free dependence data up to this discrete ambiguity.
\end{theorem}

\begin{proof}
Write a general two-qubit state in the standard Fano--Bloch form \cite{Fano1983,SchlienzMahler1995},
\begin{equation}
\rho
=
\frac14
\left[
I\otimes I
+
(\mathbf a\cdot\boldsymbol\sigma)\otimes I
+
I\otimes(\mathbf b\cdot\boldsymbol\sigma)
+
\sum_{i,j=1}^{3}
T_{ij}\,
\sigma_i\otimes\sigma_j
\right].
\label{eq:fano}
\end{equation}
Here $\mathbf a$ and $\mathbf b$ are the Bloch vectors of the reduced states and
\begin{equation}
T_{ij}
=
\Tr\!\left(
\rho\,\sigma_i\otimes\sigma_j
\right)
\end{equation}
is the two-body correlation matrix. Since
$\{I,\sigma_x,\sigma_y,\sigma_z\}^{\otimes2}$
is a basis of the two-qubit operator space, the data
$(\mathbf a,\mathbf b,T)$ determine $\rho$ uniquely.

Choose arbitrary local Bloch-sphere directions
\begin{equation}
\mathbf x,\mathbf y\in S^2.
\end{equation}
The corresponding binary projectors are
\begin{equation}
\Pi_s^{(\mathbf x)}
=
\frac12
\left(
I+s\,\mathbf x\cdot\boldsymbol\sigma
\right),
\qquad
\Pi_t^{(\mathbf y)}
=
\frac12
\left(
I+t\,\mathbf y\cdot\boldsymbol\sigma
\right),
\end{equation}
with $s,t=\pm1$. The Born rule gives
\begin{equation}
p_{st}^{\rho}(\mathbf x,\mathbf y)
=
\frac14
\left(
1+sA_\rho+tB_\rho+stC_\rho
\right),
\label{eq:born}
\end{equation}
where
\begin{equation}
A_\rho
=
\mathbf a_\rho\cdot\mathbf x,
\qquad
B_\rho
=
\mathbf b_\rho\cdot\mathbf y,
\qquad
C_\rho
=
\mathbf x^TT_\rho\mathbf y.
\label{eq:ABC}
\end{equation}

We first note that a non-product two-qubit state necessarily has full-rank reduced states. Indeed, if
\begin{equation}
\rho_A=|a\rangle\langle a|
\end{equation}
had rank one, then
\begin{equation}
\Tr\!\left[
\rho\,
\bigl(
(I-|a\rangle\langle a|)\otimes I
\bigr)
\right]
=0.
\end{equation}
Since $\rho\ge0$, its support would then lie entirely in
\begin{equation}
|a\rangle\otimes\mathcal H_B,
\end{equation}
and hence
\begin{equation}
\rho
=
|a\rangle\langle a|
\otimes\rho_B,
\end{equation}
contrary to the assumption that $\rho$ is non-product. Thus
\begin{equation}
\rho_A>0,
\qquad
\rho_B>0.
\label{eq:reduced-positive}
\end{equation}
The same holds for $\sigma$.

Consequently, in every local measurement context all row and column marginals of the Born table are strictly positive, although individual entries may vanish.

Introduce the connected correlation
\begin{equation}
\Gamma_\rho(\mathbf x,\mathbf y)
:=
C_\rho-A_\rho B_\rho.
\label{eq:gamma}
\end{equation}
It is precisely
\begin{align}
\Gamma_\rho(\mathbf x,\mathbf y)
={}&
\left\langle
(\mathbf x\cdot\boldsymbol\sigma)
\otimes
(\mathbf y\cdot\boldsymbol\sigma)
\right\rangle_\rho
\nonumber\\
&-
\left\langle
(\mathbf x\cdot\boldsymbol\sigma)\otimes I
\right\rangle_\rho
\left\langle
I\otimes
(\mathbf y\cdot\boldsymbol\sigma)
\right\rangle_\rho,
\end{align}
and is bilinear in the measurement directions:
\begin{equation}
\Gamma_\rho(\mathbf x,\mathbf y)
=
\mathbf x^T
\left(
T_\rho-\mathbf a_\rho\mathbf b_\rho^T
\right)
\mathbf y.
\label{eq:gamma-bilinear}
\end{equation}

We also define
\begin{equation}
D_\rho(\mathbf x,\mathbf y)
:=
1+C_\rho^2-A_\rho^2-B_\rho^2.
\label{eq:D}
\end{equation}
Using Eq.~\eqref{eq:born},
\begin{equation}
D_\rho
=
8\left(
p_{++}^{\rho}p_{--}^{\rho}
+
p_{+-}^{\rho}p_{-+}^{\rho}
\right).
\label{eq:Dprob}
\end{equation}
If $D_\rho=0$, both products on the right must vanish. For a nonnegative $2\times2$ table this would force at least one complete row or one complete column to vanish, contradicting Eq.~\eqref{eq:reduced-positive}. Hence
\begin{equation}
D_\rho(\mathbf x,\mathbf y)>0
\qquad
\forall\,\mathbf x,\mathbf y.
\label{eq:Dpositive}
\end{equation}
The same argument applies to $\sigma$.

Equation~\eqref{eq:Qdef} can now be evaluated directly from Eq.~\eqref{eq:born}. One finds
\begin{equation}
Q_\rho(\mathbf x,\mathbf y)
=
\frac{2\Gamma_\rho(\mathbf x,\mathbf y)}
{D_\rho(\mathbf x,\mathbf y)}.
\label{eq:QGD}
\end{equation}
Thus the all-context equality \eqref{eq:allcontext-Q} is equivalent to
\begin{equation}
\Gamma_\rho D_\sigma
=
\Gamma_\sigma D_\rho.
\label{eq:cross}
\end{equation}

Since $D_\rho,D_\sigma>0$, Eq.~\eqref{eq:cross} implies
\begin{equation}
\Gamma_\rho(\mathbf x,\mathbf y)=0
\quad\Longleftrightarrow\quad
\Gamma_\sigma(\mathbf x,\mathbf y)=0.
\label{eq:samezeros}
\end{equation}
Two nonzero bilinear forms with the same zero set are proportional.
\footnote{
For fixed $\mathbf x$, the two linear functionals of $\mathbf y$ have the same kernel and are therefore proportional. If the bilinear form has rank at least two, choose two values of $\mathbf x$ whose images are linearly independent; considering their sum then forces the proportionality constants to agree, and hence to be independent of $\mathbf x$. In the rank-one case the bilinear forms factorize, and equality of their zero sets fixes the two factors separately up to nonzero scalars. Thus the two bilinear forms differ only by one global constant.
}
Since neither $\rho$ nor $\sigma$ is a product state, neither $\Gamma_\rho$ nor $\Gamma_\sigma$ is identically zero. Hence there exists a nonzero constant $\lambda$ such that
\begin{equation}
\Gamma_\sigma(\mathbf x,\mathbf y)
=
\lambda\,
\Gamma_\rho(\mathbf x,\mathbf y)
\qquad
\forall\,\mathbf x,\mathbf y.
\label{eq:lambda-gamma}
\end{equation}

Substituting Eq.~\eqref{eq:lambda-gamma} into Eq.~\eqref{eq:cross} gives
\begin{equation}
\Gamma_\rho
\left(
D_\sigma-\lambda D_\rho
\right)
=
0.
\label{eq:lambda-pre}
\end{equation}
Thus
\begin{equation}
D_\sigma=\lambda D_\rho
\end{equation}
wherever $\Gamma_\rho\neq0$. Since $\Gamma_\rho$ is a nonzero bilinear form, its nonzero set is dense in $S^2\times S^2$. Both sides are continuous in the measurement directions, so
\begin{equation}
D_\sigma
=
\lambda D_\rho
\qquad
\forall\,\mathbf x,\mathbf y.
\label{eq:lambda-factor}
\end{equation}
Since $D_\rho,D_\sigma>0$, we have
\begin{equation}
\lambda>0.
\end{equation}

We now show that in fact $\lambda=1$. Choose a unit vector $\mathbf x_0$ satisfying
\begin{equation}
\mathbf a_\rho\cdot\mathbf x_0
=
\mathbf a_\sigma\cdot\mathbf x_0
=
0.
\label{eq:x0}
\end{equation}
Such a direction always exists in the three-dimensional Bloch space.

First suppose that the linear functional
\begin{equation}
\mathbf y
\longmapsto
\Gamma_\rho(\mathbf x_0,\mathbf y)
\end{equation}
is not identically zero. Its kernel is a two-dimensional plane. For any unit vector $\mathbf y$ in this plane,
\begin{equation}
A_\rho=A_\sigma=0,
\qquad
\Gamma_\rho=\Gamma_\sigma=0.
\end{equation}
Since $\Gamma=C-AB$, this gives
\begin{equation}
C_\rho=C_\sigma=0.
\end{equation}
Equation~\eqref{eq:lambda-factor} therefore reduces to
\begin{equation}
(\mathbf b_\sigma\cdot\mathbf y)^2
-
\lambda
(\mathbf b_\rho\cdot\mathbf y)^2
=
1-\lambda.
\label{eq:planequad}
\end{equation}

If $0<\lambda<1$, the quadratic form on the left would be positive definite on this two-dimensional plane. But it is the difference of two rank-one positive semidefinite quadratic forms and therefore has positive inertia index at most one. This is impossible. If $\lambda>1$, the same argument applied to the negative inertia index gives the same contradiction. Hence
\begin{equation}
\lambda=1.
\end{equation}

If instead
\begin{equation}
\Gamma_\rho(\mathbf x_0,\mathbf y)\equiv0,
\end{equation}
then Eq.~\eqref{eq:lambda-factor} gives Eq.~\eqref{eq:planequad} for every unit vector $\mathbf y$. If $\lambda\neq1$, this would imply
\begin{equation}
\mathbf b_\sigma\mathbf b_\sigma^T
-
\lambda\mathbf b_\rho\mathbf b_\rho^T
=
(1-\lambda)I.
\end{equation}
The left-hand side has rank at most two, whereas the right-hand side has rank three. This is impossible. Therefore
\begin{equation}
\lambda=1.
\label{eq:lambdaone}
\end{equation}

Thus
\begin{equation}
\Gamma_\sigma(\mathbf x,\mathbf y)
=
\Gamma_\rho(\mathbf x,\mathbf y)
\qquad
\forall\,\mathbf x,\mathbf y,
\label{eq:gammaequal}
\end{equation}
and Eq.~\eqref{eq:lambda-factor} becomes
\begin{equation}
C_\sigma^2-A_\sigma^2-B_\sigma^2
=
C_\rho^2-A_\rho^2-B_\rho^2.
\label{eq:squareeq}
\end{equation}

Choose again $\mathbf x_0$ satisfying Eq.~\eqref{eq:x0}, now with arbitrary $\mathbf y$. Since
\begin{equation}
A_\rho=A_\sigma=0
\end{equation}
and $\Gamma_\rho=\Gamma_\sigma$, we have
\begin{equation}
C_\rho=C_\sigma.
\end{equation}
Equation~\eqref{eq:squareeq} therefore gives
\begin{equation}
(\mathbf b_\sigma\cdot\mathbf y)^2
=
(\mathbf b_\rho\cdot\mathbf y)^2
\qquad
\forall\,\mathbf y,
\end{equation}
and hence
\begin{equation}
\mathbf b_\sigma
=
\pm\mathbf b_\rho.
\label{eq:bsign}
\end{equation}
Exchanging the two subsystems similarly gives
\begin{equation}
\mathbf a_\sigma
=
\pm\mathbf a_\rho.
\label{eq:asign}
\end{equation}

On the other hand, Eqs.~\eqref{eq:gammaequal} and \eqref{eq:gamma-bilinear} imply
\begin{equation}
T_\sigma
-
\mathbf a_\sigma\mathbf b_\sigma^T
=
T_\rho
-
\mathbf a_\rho\mathbf b_\rho^T.
\label{eq:Tconnected}
\end{equation}
Writing
\begin{equation}
\mathbf a_\sigma
=
\varepsilon_A\mathbf a_\rho,
\qquad
\mathbf b_\sigma
=
\varepsilon_B\mathbf b_\rho,
\qquad
\varepsilon_A,\varepsilon_B=\pm1,
\end{equation}
we obtain
\begin{equation}
T_\sigma
=
T_\rho
+
(\varepsilon_A\varepsilon_B-1)
\mathbf a_\rho\mathbf b_\rho^T.
\label{eq:Tsigma}
\end{equation}

Suppose that the two nonzero Bloch vectors acquire opposite signs,
\begin{equation}
\varepsilon_A\varepsilon_B=-1.
\end{equation}
Substitution of Eq.~\eqref{eq:Tsigma} into Eq.~\eqref{eq:squareeq} gives
\begin{equation}
\Gamma_\rho(\mathbf x,\mathbf y)
(\mathbf a_\rho\cdot\mathbf x)
(\mathbf b_\rho\cdot\mathbf y)
=
0
\qquad
\forall\,\mathbf x,\mathbf y.
\label{eq:mixedsign}
\end{equation}
For a non-product state $\Gamma_\rho$ is not identically zero. If both $\mathbf a_\rho$ and $\mathbf b_\rho$ are nonzero, the nonzero sets of the three factors in Eq.~\eqref{eq:mixedsign} are open and dense, so one can choose $\mathbf x,\mathbf y$ for which all three factors are nonzero, a contradiction. If one Bloch vector vanishes, its sign is immaterial and no additional state-level branch arises.

Thus only
\begin{equation}
\varepsilon_A=\varepsilon_B=1
\qquad\text{or}\qquad
\varepsilon_A=\varepsilon_B=-1
\end{equation}
remain. The first gives
\begin{equation}
(\mathbf a_\sigma,\mathbf b_\sigma,T_\sigma)
=
(\mathbf a_\rho,\mathbf b_\rho,T_\rho),
\end{equation}
and hence
\begin{equation}
\sigma=\rho.
\end{equation}
The second gives
\begin{equation}
(\mathbf a_\sigma,\mathbf b_\sigma,T_\sigma)
=
(-\mathbf a_\rho,-\mathbf b_\rho,T_\rho),
\end{equation}
which is precisely the double spin flip,
\begin{equation}
\sigma=\widetilde\rho.
\end{equation}
\end{proof}

This result suggests a particular way of identifying a quantum state. In any fixed measurement context, discarding the marginals of the Born distribution certainly loses information. The theorem shows, however, that information lost context by context need not remain lost after all contexts are considered together. The probability distributions in different contexts must all arise from the same quantum state; it is precisely this cross-context compatibility that allows a family of margin-free dependence structures to constrain, and in the above sense reconstruct, the original state.

Product states provide the natural degenerate boundary of this reconstruction problem. If
\begin{equation}
\rho=\rho_A\otimes\rho_B,
\end{equation}
then the Born probabilities factorize in every local context,
\begin{equation}
p_{st}^{\rho}(\mathbf x,\mathbf y)
=
p_s^A(\mathbf x)\,
p_t^B(\mathbf y),
\end{equation}
and therefore
\begin{equation}
p_{++}^{\rho}p_{--}^{\rho}
=
p_{+-}^{\rho}p_{-+}^{\rho}
\qquad
\forall\,\mathbf x,\mathbf y.
\label{eq:product-cross}
\end{equation}
Conversely, if Eq.~\eqref{eq:product-cross} holds in every local context, then Eq.~\eqref{eq:QGD} gives
\begin{equation}
\Gamma_\rho\equiv0,
\end{equation}
and hence
\begin{equation}
T_\rho
=
\mathbf a_\rho\mathbf b_\rho^T.
\end{equation}
The Fano--Bloch expansion then factorizes as
\begin{equation}
\rho
=
\frac12
\left(
I+\mathbf a_\rho\cdot\boldsymbol\sigma
\right)
\otimes
\frac12
\left(
I+\mathbf b_\rho\cdot\boldsymbol\sigma
\right)
=
\rho_A\otimes\rho_B.
\end{equation}
Thus the absence of nontrivial binary dependence in all local contexts is equivalent to the state being product.

As another simple case, suppose that both reduced states are maximally mixed,
\begin{equation}
\mathbf a=\mathbf b=0.
\end{equation}
Every Born table then already has uniform marginals,
\begin{equation}
p_{st}
=
\frac14
\left(
1+st\,\mathbf x^TT\mathbf y
\right).
\end{equation}
In this case,
\begin{equation}
D=1+C^2,
\qquad
\Gamma=C,
\end{equation}
and therefore
\begin{equation}
Q_\rho(\mathbf x,\mathbf y)
=
\frac{2\,\mathbf x^TT\mathbf y}
{1+(\mathbf x^TT\mathbf y)^2}.
\label{eq:QLMM}
\end{equation}
Since the map
\begin{equation}
c\longmapsto\frac{2c}{1+c^2}
\end{equation}
is one-to-one on $[-1,1]$, the all-context $Q$ data directly determine
\begin{equation}
\mathbf x^TT\mathbf y
\end{equation}
for every pair of local directions, and hence recover the full correlation matrix $T$.

\section{Copula entropy and quantum discord}
\label{sec:discord}

The preceding section retained the complete dependence nucleus in each Born-rule measurement context and considered the information contained in the whole family of such structures. We now consider a coarser description: rather than keeping the full dependence structure in each context, we compress the amount of dependence to a single number.

The starting point remains the same. Given an $N$-partite quantum state $\rho$ and a collection of local measurement bases
\begin{equation}
B=(\mathcal B_1,\ldots,\mathcal B_N),
\end{equation}
the Born rule produces a classical joint distribution
\begin{equation}
p_{\rho,B}(i_1,\ldots,i_N)
=\Tr\!\left[
\rho\,
\Pi_{i_1}^{(1)}\otimes\cdots\otimes\Pi_{i_N}^{(N)}
\right].
\end{equation}
Once $B$ is fixed, the dependence analysis of this table is an ordinary classical probability problem. We may therefore first ask: how much dependence is contained in the classical statistical table associated with this one context?

For continuous random variables, the copula entropy is defined by \cite{MaSun2011}
\begin{equation}
H_c(C)=-\int_{[0,1]^N}c(\mathbf u)\log c(\mathbf u)\,\dd^N\mathbf u,
\label{eq:copulaentropy}
\end{equation}
where $c$ is the copula density. A standard identity gives
\begin{equation}
-H_c(C)=\sum_{k=1}^{N}H(X_k)-H(X_1,\ldots,X_N).
\label{eq:ce-mi}
\end{equation}
The right-hand side is the classical total correlation, also called multi-information. Writing
\begin{equation}
\Tcl(p)=\sum_{k=1}^{N}H(p_k)-H(p),
\end{equation}
Eq.~\eqref{eq:ce-mi} becomes
\begin{equation}
\Tcl(p)=-H_c(C).
\label{eq:tcl-hc}
\end{equation}
Thus copula entropy does not define a scalar independent of total correlation. Its significance here is structural: Eq.~\eqref{eq:tcl-hc} interprets the same quantity entirely as an entropy of the margin-free dependence structure itself.\footnote{Born probabilities in finite-dimensional quantum systems are discrete. Because the discrete Sklar copula is not unique, we use $\Tcl(p)$ directly as the scalar dependence information in this section. One may alternatively adopt the standard multilinear (checkerboard) copula extension; in that representation, $-H_c$ recovers the classical total correlation exactly. See Ref.~\cite{GenestNeslehovaRemillard2017}.}

For a fixed measurement context $B$, we therefore obtain $\Tcl(p_{\rho,B})$. This quantity generally depends on the arbitrary choice of local measurement bases. If we wish to obtain a scalar associated with the quantum state itself, the most direct way to remove this context dependence is to optimize over all local contexts:
\begin{equation}
C_{\max}(\rho)=\max_B\Tcl\!\left(p_{\rho,B}\right).
\label{eq:Cmax}
\end{equation}
It is the largest classical total correlation that can be made manifest in a single local measurement context.

The quantum total correlation of the state itself is
\begin{equation}
\Tq(\rho)
=\sum_{k=1}^{N}S(\rho_k)-S(\rho)
=D\!\left(\rho\middle\|\bigotimes_{k=1}^{N}\rho_k\right).
\label{eq:Tq}
\end{equation}
The difference
\begin{equation}
Q(\rho)=\Tq(\rho)-C_{\max}(\rho)
\label{eq:Qdef2}
\end{equation}
has a direct interpretation: it is the difference between the state's total quantum correlation and the largest classical total correlation visible in any single local measurement context.

This definition does not produce a new quantity. In fact, it is precisely the familiar global quantum discord of Rulli and Sarandy \cite{RulliSarandy2011}. Let
\begin{equation}
\Phi_B(\rho)
=\sum_{\mathbf i}
\Pi_{\mathbf i}^{(B)}\rho\Pi_{\mathbf i}^{(B)}
\label{eq:dephase}
\end{equation}
be the nonselective local projective measurement in the product basis $B$, equivalently complete dephasing in that basis. It removes all off-diagonal matrix elements of $\rho$ in $B$ and leaves
\begin{equation}
\Phi_B(\rho)
=\sum_{\mathbf i}p_{\rho,B}(\mathbf i)\,\Pi_{\mathbf i}^{(B)}.
\label{eq:dephasedtable}
\end{equation}
Its eigenvalues are therefore precisely the Born probabilities, and hence
\begin{equation}
\Tq\!\left(\Phi_B(\rho)\right)
=\Tcl\!\left(p_{\rho,B}\right).
\label{eq:Tdephase}
\end{equation}
The global quantum discord can be written as \cite{RulliSarandy2011}
\begin{equation}
\GQD(\rho)
=\min_B\left[
\Tq(\rho)-\Tq\!\left(\Phi_B(\rho)\right)
\right].
\label{eq:GQDdef}
\end{equation}
Therefore
\begin{equation}
\GQD(\rho)
=\Tq(\rho)-\max_B\Tcl\!\left(p_{\rho,B}\right),
\label{eq:GQDresult}
\end{equation}
which is exactly Eq.~\eqref{eq:Qdef2}.

From the present viewpoint, the appearance of global quantum discord is quite natural. We do not begin with a pre-existing definition of quantum correlation and then seek a copula representation of it. Instead, we perform a sequence of elementary steps: use the Born rule to obtain classical joint statistics, separate marginals from dependence, and then remove the arbitrariness of the local measurement context. The scalar obtained along this route happens to return us to the established quantity of global quantum discord.

At the same time, Section~\ref{sec:reconstruction} shows that this scalar is not the full object we started from. A quantum state first determines the context-indexed family
\begin{equation}
\mathfrak D_\rho
=\left\{\mathcal D\!\left(p_{\rho,B}\right)\right\}_B,
\end{equation}
whereas $\max_B\Tcl(p_{\rho,B})$ is only one scalar invariant extracted from this family. Just as a function generally contains more information than its maximum value, the full all-context dependence family is much richer than this optimization. The reconstruction theorem of Section~\ref{sec:reconstruction} provides a concrete illustration of this distinction.

\section{Outlook}
\label{sec:outlook}

The complete reconstruction result established here concerns only two qubits. From the viewpoint of this paper, two extensions are particularly natural.

\subsection{Multiqubit systems}
\label{subsec:multiqubit}

The conceptual extension to multipartite systems is nearly immediate. This directness is worth noting. A correlation construction that is well defined for two parties need not, by itself, specify a unique multipartite continuation; additional structural choices can enter. By contrast, the starting point of the present approach has no analogous ambiguity at this stage: once a local measurement context is chosen for $N$ quantum subsystems, the Born rule produces an ordinary $N$-variable joint probability distribution. The passage from bivariate to multivariate dependence is already part of classical probability theory. Thus the marginal--dependence separation revealed by Sklar's theorem does not suffer a conceptual discontinuity when one passes from two variables to many \cite{Sklar1959,Nelsen2006}. The genuinely new difficulty is instead whether these multivariate dependence structures, taken over all quantum measurement contexts, still suffice to identify the underlying quantum state.

For multiqubit systems, this extension can be pursued along the dependence-nucleus viewpoint already adopted in Section~\ref{subsec:binary}. Following Geenens, we characterize margin-free dependence in a discrete probability table through invariance under positive local reweightings; the odds ratio of the binary $2\times2$ case is the simplest coordinate of this structure \cite{Geenens2020}. For general multivariate binary tables, the standard log-linear interaction parameters of multiway contingency-table theory provide a natural technical language for organizing the same type of discrete dependence structure \cite{Agresti2013,DarrochLauritzenSpeed1980}.

Our preliminary investigations suggest that this language also exposes a structure particularly suited to quantum reconstruction: under conditioning, higher-order $N$-body dependence data naturally induce conditional $(N-1)$-body dependence data. In quantum mechanics, fixing a local measurement outcome likewise produces a conditional state on the remaining subsystems. This points to a concrete recursive route to all-context dependence reconstruction for multiqubit systems. The full problem, including the compatibility of the resulting conditional states and the weakest conditions under which the global state can be recovered, is nontrivial and will be left for future work.

\subsection{Continuous-variable systems}
\label{subsec:cv}

Another natural direction is continuous-variable quantum systems. From the statistical viewpoint adopted here, the continuous case has an immediate advantage: when the marginal distributions are continuous, the copula in Sklar's theorem is uniquely determined by the joint distribution. Thus, in each fixed measurement context, one no longer needs an additional object such as the dependence nucleus to represent margin-free dependence; the copula itself can play this role directly.

The subtlety shifts to the quantum side. For qubits, we take the local contexts to be all local binary projective measurements, a family with a simple Bloch-sphere parametrization. For a general continuous-variable system, there is no equally immediate unique choice of what should count as the corresponding family of ``all local contexts.''

A natural and standard first step is to consider a collection of bosonic modes and restrict the local contexts to quadrature measurements. For the $i$th mode, let $(\widehat X_i,\widehat P_i)$ be a canonical pair of quadratures and define the rotated quadrature
\begin{equation}
\widehat X_{\theta_i}^{(i)}
=\widehat X_i\cos\theta_i+\widehat P_i\sin\theta_i.
\end{equation}
The position-like and momentum-like quadratures are the special cases $\theta_i=0$ and $\theta_i=\pi/2$. For $N$ modes, a set of local angles
\begin{equation}
\boldsymbol\theta=(\theta_1,\ldots,\theta_N)
\end{equation}
defines a local measurement context and, through the Born rule, a continuous joint density
\begin{equation}
p_{\rho,\boldsymbol\theta}(x_1,\ldots,x_N).
\end{equation}
This choice is motivated by homodyne tomography. For a single mode, Vogel and Risken showed that the full set of rotated-quadrature probability distributions suffices to reconstruct the quantum state by phase-space tomography \cite{VogelRisken1989}; see also Ref.~\cite{LvovskyRaymer2009}. The corresponding multimode construction uses the joint local quadrature statistics.

The new question from the present viewpoint is therefore not whether the quadrature statistics are sufficient for state reconstruction, but what remains if, in every quadrature context, all local marginals are discarded and only the corresponding copula is kept. For two bosonic modes, for example, one may consider
\begin{equation}
\rho_{AB}\longmapsto
\left\{C_{\rho;\theta_A,\theta_B}\right\}_{\theta_A,\theta_B},
\label{eq:cvmap}
\end{equation}
where $C_{\rho;\theta_A,\theta_B}$ is the unique copula of the joint quadrature distribution
\begin{equation}
p_\rho(x_A,x_B\mid\theta_A,\theta_B).
\end{equation}
The continuous-variable reconstruction problem is then whether this all-context copula family still identifies the original state, or equivalently, how much state information remains in the complete cross-context dependence structure after the local marginal information has been discarded context by context.

The problem is structurally close to the two-qubit setting studied here, but it also exhibits an interesting reversal: the dependence object in a single context is cleaner because the copula is unique, while the quantum side requires a further choice of which family of contexts should be regarded as natural. Quadrature measurements provide a particularly natural starting point because of their established tomographic role. Whether one should enlarge the context family to more general local measurements, and whether doing so changes the state-separating power of the all-context dependence family, remain open questions.

\section{Conclusion}

We have applied the marginal--dependence separation revealed by Sklar's theorem to Born statistics across local measurement contexts. For two qubits, although marginal information is discarded separately in every context, the full family of dependence nuclei reconstructs any non-product state up to at most a double spin flip; product states form the natural relationally degenerate case. When the dependence in each context is compressed to total correlation and optimized over local contexts, the resulting quantity is precisely global quantum discord. Thus, information lost in each individual statistical manifestation need not remain lost when the entire compatible family of manifestations is considered.

From a broader perspective, the idea of identifying a composite quantum state 
from relational data among its subsystems already has several important 
precedents. For example, Mermin showed that a complete set of subsystem 
correlations determines the density operator of the composite system 
\cite{Mermin1998}. A related question appears in the literature on quantum 
self-testing, where one asks whether observed correlations can uniquely 
determine the underlying quantum state and measurements 
\cite{SupicBowles2020}. The question considered here is different: even after 
each Born table is individually stripped of its marginals, does the entire 
context-indexed dependence family still determine $\rho$?

We also note that a notion of quantum copula has previously been proposed 
\cite{LovasAndai2019}. There, the construction defines a quantum analogue of a 
copula directly at the level of the quantum state. This is distinct from the 
context-wise Born-statistical construction considered here.

\paragraph*{Acknowledgement}
This work is supported by the NSFC Grant No.1250050230.

\appendix
\section{Sklar's theorem and copulas}
\label{app:sklar}

Copula theory asks how the part of a joint probability distribution determined by the individual random variables can be separated from the part that describes how the variables depend on one another. We briefly review the facts needed in the main text. Standard references are Sklar's original paper and Nelsen's monograph \cite{Sklar1959,Nelsen2006}.

\subsection{Joint and marginal distributions}

Let
\begin{equation}
\mathbf X=(X_1,\ldots,X_N)
\end{equation}
be $N$ real-valued random variables with joint cumulative distribution function
\begin{equation}
F(x_1,\ldots,x_N)
=\Pr(X_1\le x_1,\ldots,X_N\le x_N).
\end{equation}
The marginal cumulative distribution function of the $i$th variable is
\begin{equation}
F_i(x_i)=\Pr(X_i\le x_i).
\end{equation}
The joint distribution determines all marginals, but the converse is false: even when every $F_i$ is fixed, the way in which the variables vary together is not determined. Copula theory is designed to represent this remaining structure; when the marginals are continuous, it does so uniquely.

\subsection{Sklar's theorem}

Sklar's theorem gives the exact form of this separation. For any $N$-dimensional joint distribution $F$, there exists an $N$-dimensional copula $C$ such that
\begin{equation}
F(x_1,\ldots,x_N)
=C\!\left(F_1(x_1),\ldots,F_N(x_N)\right).
\label{eq:sklar}
\end{equation}
Here
\begin{equation}
C:[0,1]^N\to[0,1]
\end{equation}
is itself a joint distribution function whose one-dimensional marginals are all uniform on $[0,1]$. Conversely, given one-dimensional distribution functions $F_1,\ldots,F_N$ and any $N$-copula $C$, Eq.~\eqref{eq:sklar} defines a valid joint distribution with those marginals.

The crucial point for the present paper is uniqueness: if all marginal distributions $F_i$ are continuous, then $C$ is unique. In the continuous case one may therefore write, in a precise sense,
\begin{equation}
\text{joint distribution}=\text{marginals}+\text{copula},
\end{equation}
where the marginals describe how the variables are distributed individually and the copula describes how they depend on one another.

\subsection{Uniformization and the meaning of a copula}

In the continuous case, the same separation can be seen directly through the probability-integral transform. Define
\begin{equation}
U_i=F_i(X_i).
\end{equation}
If $F_i$ is continuous, then
\begin{equation}
U_i\sim\mathrm{Uniform}(0,1).
\end{equation}
Thus the transformation
\begin{equation}
\mathbf X=(X_1,\ldots,X_N)
\longmapsto
\mathbf U=(U_1,\ldots,U_N)
\end{equation}
removes the individual shape of every marginal, since all $U_i$ have the same uniform distribution. The dependence among the variables is not removed. Indeed,
\begin{equation}
\Pr(U_1\le u_1,\ldots,U_N\le u_N)
=C(u_1,\ldots,u_N).
\end{equation}
The copula can therefore be understood directly as the joint distribution that remains after all marginals have been uniformized.

This also explains why the copula is insensitive to the individual scales of the variables. Separate strictly increasing reparametrizations of the variables do not change their copula.

If the variables are independent,
\begin{equation}
F(x_1,\ldots,x_N)=\prod_{i=1}^{N}F_i(x_i),
\end{equation}
and the corresponding independence copula is
\begin{equation}
C_{\perp}(u_1,\ldots,u_N)=\prod_{i=1}^{N}u_i.
\label{eq:indcopula}
\end{equation}
Thus nontrivial dependence is encoded by departure from $C_{\perp}$.

\subsection{Copula density}

Suppose the joint distribution and all marginals are absolutely continuous, with densities
\begin{equation}
p(x_1,\ldots,x_N),
\qquad
p_i(x_i),
\end{equation}
and suppose the copula also has a density
\begin{equation}
c(u_1,\ldots,u_N)
=\frac{\partial^N C}{\partial u_1\cdots\partial u_N}.
\end{equation}
Differentiating Eq.~\eqref{eq:sklar} gives
\begin{equation}
p(x_1,\ldots,x_N)
=c\!\left(F_1(x_1),\ldots,F_N(x_N)\right)
\prod_{i=1}^{N}p_i(x_i).
\label{eq:copuladensity}
\end{equation}
In particular, $c=1$ almost everywhere corresponds to independence. Equation~\eqref{eq:copuladensity} makes the marginal--dependence separation especially transparent:
\begin{equation}
\text{joint density}=\text{copula density}\times\text{marginal densities}.
\end{equation}

\subsection{Copula entropy and total correlation}

The copula entropy used in Section~\ref{sec:discord} is, for a continuous copula with density $c(\mathbf u)$, defined by
\begin{equation}
H_c=-\int_{[0,1]^N}c(\mathbf u)\log c(\mathbf u)\,\dd^N\mathbf u.
\label{eq:appHc}
\end{equation}
The classical total correlation of the random vector $\mathbf X$ is
\begin{equation}
\Tcl(\mathbf X)
=\KL\!\left(
p(x_1,\ldots,x_N)
\middle\|
\prod_i p_i(x_i)
\right),
\end{equation}
that is,
\begin{equation}
\Tcl
=\int p(\mathbf x)
\log\frac{p(\mathbf x)}{\prod_i p_i(x_i)}\,
\dd^N\mathbf x.
\label{eq:appTC}
\end{equation}
Using Eq.~\eqref{eq:copuladensity},
\begin{equation}
\frac{p(\mathbf x)}{\prod_i p_i(x_i)}=c(\mathbf u),
\end{equation}
and under the change of variables $u_i=F_i(x_i)$,
\begin{equation}
\dd^N\mathbf u
=\prod_i p_i(x_i)\,\dd^N\mathbf x.
\end{equation}
Therefore
\begin{align}
\Tcl
&=\int_{[0,1]^N}c(\mathbf u)\log c(\mathbf u)\,\dd^N\mathbf u\\
&=-H_c,
\end{align}
and hence
\begin{equation}
\Tcl=-H_c.
\label{eq:appTCHc}
\end{equation}
This is the relation between copula entropy and classical total correlation used in the main text. Ma and Sun expressed the same identity as ``mutual information is negative copula entropy'' \cite{MaSun2011}; in the multivariate setting, we use the less ambiguous terminology total correlation or multi-information.

\subsection{The discrete case}

Sklar's theorem itself is not restricted to continuous variables. For a discrete joint distribution, there still exists a copula $C$ satisfying Eq.~\eqref{eq:sklar}. The difference is that the range of each $F_i$ covers only a discrete subset of $[0,1]$, so the joint distribution fixes $C$ only on
\begin{equation}
\operatorname{Ran}(F_1)\times\cdots\times\operatorname{Ran}(F_N).
\end{equation}
Extending the resulting subcopula to all of $[0,1]^N$ is generally nonunique. Consequently, the copula is not uniquely determined by the original discrete joint distribution.

This is the distinction between the continuous and discrete cases that matters most for the present paper. In the continuous case,
\begin{equation}
\text{joint distribution}\longrightarrow\text{unique copula},
\end{equation}
whereas in the discrete case an ordinary Sklar copula cannot by itself serve as a unique intrinsic margin-free dependence object. The difficulties and limitations of discrete copulas in this respect are discussed systematically by Genest and Ne\v{s}lehov\'a \cite{GenestNeslehova2007}.

Accordingly, for the binary Born statistics studied in the main text, we do not select an arbitrary representative from the nonunique family of Sklar copulas. Instead, we adopt the dependence-nucleus viewpoint of Geenens \cite{Geenens2020}, using the equivalence class under positive local reweightings as the margin-free dependence object. The positive-table odds-ratio description, the uniform-margin representative, and the two-qubit reconstruction theorem in Section~\ref{sec:reconstruction} are all built on this discrete implementation.

\bibliographystyle{unsrt}
\bibliography{references}

\end{document}